\documentclass[11pt]{article}
\usepackage{graphicx} 
\usepackage{lineno}
\usepackage{enumerate}
\usepackage{enumitem}
\usepackage{algorithm}
\usepackage{appendix}
\usepackage[noend]{algpseudocode}
\usepackage{amsmath, amsfonts, amsthm}
\usepackage{todonotes}
\usepackage{hyperref}
\usepackage{lineno}
\usepackage{thm-restate}

\usepackage{xcolor}

\newtheorem{lemma}{Lemma}
\newtheorem{theorem}{Theorem}
\newtheorem{definition}{Definition}
\newtheorem{claim}{Claim}

\usepackage[letterpaper,top=1in,bottom=1in,left=1in,right=1in]{geometry}

\title{Simpler Optimal Sorting from a Directed Acyclic Graph}

\date{}

\begin{document}

 \author{Ivor van der Hoog, Eva Rotenberg, Daniel Rutschmann.  \\ Technical University of Denmark, DTU }

\maketitle
\begin{abstract}
    Fredman proposed in 1976 the following algorithmic problem: 
    Given are a ground set $X$, some partial order $P$ over $X$, and some comparison oracle $O_L$ that specifies a linear order $L$ over $X$ that extends $P$. A \emph{query} to $O_L$ has as input distinct $x, x' \in X$ and outputs whether $x <_L x'$ or vice versa. 
    If we denote by $e(P)$ the number of linear orders that extend $P$, then it follows from basic information theory that $\log e(P)$ is a worst-case lower bound on the number of queries needed to output the sorted order of $X$. 

    Fredman did not specify in what form the partial order is given. 
    Haeupler, Hladík, Iacono, Rozhon, Tarjan, and T\v{e}tek ('24) propose to assume as input a directed acyclic graph, $G$, with $m$ edges and $n=|X|$ vertices. Denote by $P_G$ the partial order induced by $G$.
    Their algorithmic performance is measured in running time and the number of queries used, where they use $\Theta(m + n +  \log e(P_G))$ time and $\Theta(\log e(P_G))$ queries to output $X$ in its sorted order. 
    Their algorithm is worst-case optimal in terms of running time and queries, both. 
    Their algorithm combines topological sorting with heapsort.
    Their analysis relies upon sophisticated counting arguments using entropy, recursively defined sets defined over the run of their algorithm, and vertices in the graph that they identify as bottlenecks for sorting.  

    In this paper, we do away with sophistication. 
    We show that when the input is a directed acyclic graph then the problem admits a simple solution using $\Theta(m + n + \log e(P_G))$ time and $\Theta(\log e(P_G))$ queries. 
    Especially our proofs are much simpler as we avoid the usage of advanced charging arguments and data structures,
    and instead rely upon two brief observations. 
\end{abstract}

\paragraph{Funding.}
Ivor van der Hoog received funding from the European Union's Horizon 2020 research and innovation programme under the Marie Sk\l{}odowska-Curie grant agreement No 899987.
Daniel Rutschmann and Eva Rotenberg received funding from Carlsberg Foundation
Young Researcher Fellowship CF21-0302 “Graph Algorithms with Geometric Applications”.
Eva Rotenberg additionally receives funding from 
DFF Grant 2020-2023 (9131-00044B) “Dynamic Network Analysis”, and VILLUM Foundation
grant VIL37507 “Efficient Recomputations for Changeful Problems”.

\section{Introduction}
Sorting is a fundamental problem in computer science.
In 1976, Fredman introduced a natural generalisation of sorting, namely sorting under partial information~\cite{fredman_how_1976}: 

Given a ground set $X$ of size $n$, and some partial order $P$ on $X$, and an oracle $O_L$ with access to a ground-truth linear order $L$ of $X$, the task is to minimise the number of oracle queries to recover the linear order of $X$. A query takes two input elements from $X$, and outputs their relation in the linear order. Algorithmic efficiency is measured in time, and the number of queries used. 
Let $e(P)$ denote the number of linear orders that extend the partial order $P$. Then, for any binary (e.g. yes/no) query, 
$\lceil\log_2(e(P))\rceil$ is a worst-case lower bound for the number of queries needed, simply because the sequence of answers should be able to lead to any possible extension of the order. 

\paragraph{Previous work.}
Fredman \cite[TCS'76]{fredman_how_1976} shows an exponential-time algorithm for sorting under partial information that uses $\log e(P) + O(n)$ queries.
This fails to match the $\Omega(\log e(P))$ lower bound
when $\log e(P)$ is sublinear. Kahn and Saks \cite[Order'84]{kahn_balancing_1984} prove
that there always exists a query which reduces the number of remaining linear extensions by a
constant fraction; showing that an $O(\log e(P))$-query algorithm exists. 
 Kahn and Kim \cite[STOC'92]{kahn_entropy_1992} are the first to also consider algorithmic running time.
 They note that Kahn and Saks \cite{kahn_balancing_1984} can preprocess a partial order using exponential time and space, so that given $O_L$ it can output the sorted order of $X$ in linear time plus $O(\log e(P))$ queries.
Kahn and Kim \cite[STOC'92]{kahn_entropy_1992} propose the first polynomial-time algorithm that performs $O(\log e(P))$ queries.
They do not separate preprocessing the partial order and oracle queries, and use an unspecified polynomial time using the ellipsoid method. 
Cardinal, Fiorini, Joret, Jungers and Munro~\cite[STOC'10]{cardinal_sorting_2013} (in their full version in Combinatorica) preprocess  $P$ in $O(n^{2.5})$ time, to output the sorted $X$ in $O(\log e(P) + n)$ time using $O(\log e(P))$ queries. 
Their runtime poses the question whether query-optimal (sub)quadratic-time algorithms are possible. 

For subquadratic algorithms, it becomes relevant how one obtains the partial order. 
Van der Hoog, Kostityna, L\"{o}ffler and Speckmann~\cite[SOCG'19]{van_der_hoog_preprocessing_2019} study the problem in a restricted setting where $P$ is induced by a set of intervals. Here, distinct $x_i, x_j \in X$ are incomparable whenever their intervals $([a_i, b_i], [a_j, b_j])$ intersect (otherwise, $x_i \prec_P x_j$ whenever $b_i < a_j$). 
They preprocess the intervals in $O(n \log n)$ time such that, given $O_L$, they can produce a pointer to a linked list storing $X$ in its sorted order using $O(\log e(P))$ queries and time. For their queries, they use \emph{finger search}~\cite{brodal2018finger} (i.e., \emph{exponential search} from a pointer). Finger search has as input a value $x_i$, a sorted list $\pi$, and a finger $p \in \pi$ with $p < x_i$. 
It finds the farthest $q$ along $\pi$ for which $q < x_i$ in $O(1+ \log d_i)$ time and comparisons. Here, $d_i$ denotes the length of the sublist from $p$ to $q$~\cite{brodal2018finger}.

Van der Hoog and Rutschmann~\cite[FOCS'24]{HoogRutschmann24} assume that $P$ is given as another oracle $O_P$ where queries receive distinct $x, x' \in X$ and output whether $x \prec_P x'$.
For fixed $c \geq 1$, they preprocess $O_P$ using $O(n^{1 + \frac{1}{c}})$ time and queries to $O_P$. Given $O_L$, they produce a pointer to a linked list storing $X$ in its sorted order using $O(c \cdot \log e(P))$ queries and time. 
Their query algorithm also makes use of {finger search}.
They show matching lower bounds in their setting. 

Haeupler, Hladík, Iacono, Rozhon, Tarjan, and T\v{e}tek~\cite['24]{haeupler2024fast} assume that the partial order $P_G$ is induced by a graph $G$ with $m$ edges. Their algorithm can be seen as a combination of topological sort and heapsort. 
Their algorithm first isolates a collection $B \subset X$ which they call \emph{bottlenecks}. 
They then iteratively consider all sources $S$ in $G - B$, remove the source $s \in S$ that is \emph{minimum} in the linear order $L$, and append $s$ to the output $\pi$.
Before appending $s$ to $\pi$, they find the maximum prefix $B_s = \{ b \in B \mid b <_L s \}$ using finger search where the finger is the head of $B$. They remove $B_s$ from $B$, append $B_s$ to $\pi$, and then append $s$. 
To obtain the minimum $s \in S$, they store all current sources of the graph in a \emph{heap} where comparisons in the heap use the linear oracle $O_L$.  
Their algorithm does not separate preprocessing the graph $G$, and oracle queries. It uses $\Theta(n + m + \log e(P_G))$ overall time and $
\Theta(\log e(P))$ linear oracle queries.
Since reading the input takes at least $\Omega(n + m)$ time, their overall algorithmic runtime is thereby worst-case optimal. 
The bulk of their analysis is dedicated to charging the algorithmic runtime and query time to the workings of their heap, and to handling the special bottleneck vertices.
We note that for any input $k$, their solution can also report the first $k$ elements of the sorted order in optimal time.

\paragraph{Contribution.}
We consider the setting from~\cite{haeupler2024fast}, where the input is a directed acyclic graph $G$ over $X$ with $m$ edges and where, unlike \cite{cardinal_sorting_2013, HoogRutschmann24, van_der_hoog_preprocessing_2019}, one does not separate the algorithmic performance between a pre-processing phase using $G$ and a phase that uses queries to $O_L$. 
We show that this problem formulation 
allows for a surprisingly simple solution:
$G$ denotes the input and $H_i$ the graph at iteration $i$. 
Remove a maximum-length directed path $\pi$ from $G$ to get $H_1$. 
Iteratively remove an \emph{arbitrary} source $x_i$ from $H_i$ to get $H_{i+1}$. Insert $x_i$ into $\pi$ using \emph{finger search} where the finger $p_i$ is the farthest in-neighbour of $x_i$ in $G$ along $\pi$ (we find $p_i$ by simply iterating over all in-neighbours of $x_i$).  
We use $\Theta(n + m + \log e(P_G))$ time and $\Theta(\log e(P_G))$ queries. We consider this algorithm to be simpler. 
Our proof of correctness is considerably simpler, as it relies upon only one counting argument and one geometric observation.

\section{Algorithm}

The input is a directed acyclic graph $G$ with vertex set $X = (x_1, \ldots, x_n)$. 
This graph induces a partial order $P_G$ where $x_i \prec x_j$ if and only if there exists a directed path from $x_i$ to $x_j$ in $G$. The input also contains an oracle $O_L$ that specifies a linear order $L$ over $X$. For any distinct $x_i, x_j$, an (oracle) query answers whether $x_i <_L x_j$. 
The goal is to output $X$ in its sorted order $L$.

We describe our algorithm.
Our key observation is that the problem becomes significantly easier when we first extract a maximum-length directed path $\pi$ from $G$. 
We iteratively remove an \emph{arbitrary} source $x_i$ from $H = G - \pi$ and insert it into $\pi$. Once $H$ is empty, we return $\pi$.

In our analysis, our logarithms are base $2$ and $d^*(x_i)$ denotes the sum of the out-degree and the in-degree of a vertex $x_i$ in $G$. 

\paragraph{Data structures.}
We maintain a path $\pi$ where for all vertices $p, q \in \pi$: $p <_L q$ if and only if $p$ precedes $q$ in $\pi$. 
We store the path $\pi$ in a dynamic list order structure $T_\pi$ which is a data structure that supports the following operations:  

\textsc{FingerInsert}$(q_i, x_i)$. Given vertices $x_i \not \in \pi$ and $q_i \in \pi$, insert $x_i$ into $\pi$ succeeding  $q_i$.

\textsc{Search}$(x_i, p_i, O_L)$. Given a vertex $x_i \not \in \pi$ and a vertex $p_i \in \pi$ with $p_i <_L x_i$. Return the farthest vertex $q_i$ along $\pi$ where $q_i <_L x_i$. 
Let $d_i$ denote the number of vertices on the subpath from $p_i$ to $q_i$ along $\pi$. 
We want to use $O(1 + \log d_i)$ time and queries to $O_L$.

\textsc{Compare}$(p, q)$. Given $p, q \in \pi$, return $q$ if it succeeds $p$ in $\pi$ in $O(1)$ time (return $p$ otherwise).

\noindent
To achieve this, we store $\pi$ in leaf-linked a balanced finger search tree $T_\pi$. Many implementations exist~\cite{brodal2018finger}. We choose the level-linked (2-4)-tree by Huddleston and Mehlhorn~\cite{huddleston1982new}, which supports \textsc{FingerInsert} in amortised $O(1)$-time, and \textsc{Search} in $O(1 + \log d_i)$ time. 
For \textsc{Compare}, we store the linked leaves in the simple dynamic list ordering structure by Bender et al.~\cite{bender2002two} whose simplest implementation supports \textsc{Compare} in $O(1)$ time, and \textsc{FingerInsert} in amortised $O(1)$ time. 

\begin{algorithm}[h]
\caption{Sort(directed acyclic graph $G$ over a ground set $X$, Oracle $O_L$) \hfill time}\label{algo:sort}
\begin{algorithmic}[1]
\State $\pi \gets$ a longest directed path in $G$ \Comment{$O(n + m)$}
\State $T_\pi$ $\gets$ a level-linked (2-4)-tree over $\pi$~\cite{huddleston1982new} \Comment{$O(n)$} 
\State $H \gets G - \pi$ \Comment{$O(n + m)$}
\State Compute for each vertex in $H$ its in-degree in $H$ \Comment{$O(n + m)$}
\State $S \gets $ sources in $H$ \Comment{$O(n)$}
 \While{ $S \neq \emptyset$}
\State Remove an arbitrary vertex $x_i$ from $S$ \Comment{$O(1)$} 
\State $p_i \gets$ a dummy vertex, which is prepended before the head of $\pi$ \Comment{$O(1)$}
\ForAll{in-neighbors $u$ of $x_i$ in $G$} \Comment{$O( d^*(x_i))$} 
\State $p_i \gets $ \textsc{Compare}($p_i$, $u$) \Comment{$O(1)$}
\EndFor
\State Remove $x_i$ from $H$ and add any new sources in $H$ to $S$ \Comment{$O(d^*(x_i))$}
\State $q_i \gets $ \textsc{Search}($x_i$, $p_i$, $O_L$) \Comment{$O(1 + \log d_i)$}
\State \textsc{FingerInsert}($q_i$, $x_i$)\Comment{$O(1)$ amortised}
\EndWhile
\State \textbf{return} the leaves of $T_\pi$ in order \Comment{$O(n)$}
\end{algorithmic}
\end{algorithm}

\noindent
\paragraph{Algorithm and runtime analysis.} Our algorithm is straightforward. The pseudo code is given by Algorithm~\ref{algo:sort} where each algorithmic step also indicates the running time. We require linear space. 
We first extract a longest directed path $\pi$ from $G$, which we store in our two data structures. 
This takes $O(n + m)$ total time. 
If $\pi$ has $n - k$ vertices, this leaves us with a graph $H$ with $k$ vertices. 

We then iteratively remove an arbitrary source $x_i$ from $H$. 
Since $x_i$ is a source in $H$, all vertices $p \in X$ that have a directed path to $x_i$ in $G$ must be present in $\pi$. 
We iterate over all in-neighbours of $x_i$ in $G$, and use our \textsc{Compare} operation to find the in-neighbour $p_i$ that is furthest along $\pi$. 
If $x_i$ has no in-neighbours, $p_i$ is a dummy vertex that precedes the head of $\pi$ instead. 
Finally, we proceed in a way that is very similar to the algorithm for topological sorting by Knuth~\cite{knuth1997art}: We remove $x_i$ from $H$, iterate over all out-neighbours of $x_i$ in $H$ and decrement their in-degree, and update the linked list of sources to include the newly found ones. 
Since we inspect each in- and out-edge of $x_i$ only once, this takes $O(m)$ total time.

We then want to find the farthest $q_i$ along $\pi$ that precedes $x_i$ in 
$L$.  In the special case where $p_i$ is the dummy vertex preceding the head $h$ of $\pi$, we use a query to $O_L$ to check whether $x_i <_L h$. If so, we return $q_i = p_i$. Otherwise, we set $p_i = h$ and invoke \textsc{Search}. 
This returns $q_i$ in $O(1 + \log d_i)$ time and queries, where $d_i$ is the length of the subpath from $p_i$ to $q_i$ in $\pi$ ($d_i = 1$ if $p_i = q_i$). 

Finally, we insert $x_i$ into $T_\pi$  succeeding $q_i$ in constant time. Updating the dynamic list order data structure and rebalancing $T_\pi$ takes amortised $O(1)$ time. Thus, these operations take $O(n)$ total time over all insertions. 
\textnormal{Therefore, our algorithm uses } $O\Bigl(n + m + \sum\limits_{x_i \in H} (1 + \log d_i)\Bigr)$ \textnormal{ time  and } $O\Bigl(\sum\limits_{x_i \in H} (1 + \log d_i)\Bigr)$ \textnormal{ queries.}

\subsection{Proof of optimality}
Recall that $e(P_G)$ denotes the number of linear extensions of $P_G$ and let $H$ start out with $k$ vertices. 
For ease of analysis, we re-index the vertices $X$ so that the path $\pi$ are vertices $(x_{k+1}, \ldots, x_n)$, in order, and $x_i$ for $i \in [k]$ is the $i$'th vertex our algorithm inserts into $T_\pi$.
We prove that our algorithm is tight by showing that $\sum\limits_{i = 1}^k (1 + \log d_i)  = k + \sum\limits_{i = 1}^k \log d_i \in O(\log e(P_G))$.

\begin{lemma}
\label{lem:kbound}
Let $G$ be a directed acyclic graph, $P_G$ be its induced partial order and $\pi$ be a longest directed path in $G$. If $\pi$ has  $n - k$ vertices then $\log e(P_G) \geq k$.
\end{lemma}

\begin{proof}
Denote $\ell(x)$ the length of the longest directed path in $G$ from a vertex $x$.
Denote for $i \in [n]$ by $L_i :=  | \{ x \in X \mid \ell(x) = i \} |$. 
For each $i \in [n - k]$, there is one $u \in \pi$ with $\ell(u) = i$ so: 

\begin{equation}
\label{eq:ksum}
        k = \sum\limits_{i=1}^{n - k}  (L_i -  1).
\end{equation}

Consider a linear order $L'$ that is obtained by first sorting all $x \in X$ by $\ell(x)$ from high to low, and ordering $(u, v)$ with $\ell(u) = \ell(v)$ arbitrarily. 
The order $L'$ must extend $P_G$, since any vertex $w$ that has a directed path towards a vertex $u$ must have that $\ell(w) > \ell(u)$. 
We count the number of distinct $L'$ that we obtain in this way to lower bound $e(P_G)$:

\[
e(P_G) \geq \prod_{i=1}^{n - k} L_i ! \quad \Rightarrow \quad  e(P_G) \geq  \prod_{i=1}^{n - k} 2^{L_i - 1} = 2^{\left( \sum\limits_{i=1}^{n - k} (L_i - 1) \right)} \quad  \Rightarrow \quad  e(P_G) \geq  2^k
\]

\noindent
Here, the first implication uses that $x! \geq 2^{x - 1}$ and the second implication uses Equation~\ref{eq:ksum}.
\end{proof}

What remains is to show that $\sum\limits_{i = 1}^k \log d_i \in O(\log e(P_G))$.
To this end, we create a set of intervals:

\begin{definition}
    \label{def:intervals}
     Let $\pi^*$ be the directed path that Algorithm~\ref{algo:sort} outputs. 
      For any $i \in [n]$ denote by $\pi^*(x_i)$ the index of $x_i$ in $\pi^*$. 
     We create an embedding $E$ of $X$ by placing $x_i$ at position $\pi^*(x_i)$. 

         Recall that for $i \in [k]$, $p_i$ is the finger from where Algorithm~\ref{algo:sort} invokes \textsc{Search} with $x_i$ as the argument. 
     We create as set $\mathcal{R}$ of $n$ open intervals $R_i = (a_i, b_i) \subseteq [0, n]$ as follows:
     \begin{itemize}[noitemsep, nolistsep]
         \item  If $i > k$ then $(a_i, b_i) := (\pi^*(x_i) - 1, \pi^*(x_i))$.
         \item Else, $(a_i, b_i) := (\pi^*(p_i), \pi^*(x_i) )$.
     \end{itemize}
\end{definition}

\begin{lemma}
\label{lem:disjointness}
    Given distinct $x_i, x_j \in X$. If there exists a directed path from $x_i$ to $x_j$ in $G$ then the intervals  $R_i = (a_i, b_i)$ and $R_j = (a_j, b_j)$ are disjoint with  $b_i \leq a_j$.  
\end{lemma}

\begin{proof}
We consider three cases. 

If $i > k$ and $j > k$ then $x_i$ and $x_j$ are part of the original longest directed path in $G$. 
If $x_i$ has a directed path to $x_j$ in $G$ then $x_i$ precedes $x_j$ on the original longest directed path. 
Thus,  $\pi^*(x_i)$ precedes $\pi^*(x_j)$ and $R_i$ and $R_j$ are distinct unit intervals where $R_i$ precedes $R_j$. 

If $i \leq k$ and $j > k$ then $x_i$ was inserted into the path $\pi$ with $x_j$ already in $\pi$. 
Since $\pi^*$ is a linear extension of $P_G$, it follows that $x_i$ was inserted preceding $x_j$ in $\pi$. It follows that $x_i$ precedes $x_j$ in $\pi^*$ and thus $R_i = (\pi^*(p_i), \pi^*(x_i))$ lies strictly before $R_j = ( \pi^*(x_j) - 1, \pi^*(x_j))$. 

If $j \leq k$ then the vertex $p_j$ must equal or succeed $x_i$ in the path $\pi^*$. It follows that $b_i = \pi^*(x_i) \leq  a_i =  \pi^*(p_j)$. The fact that these intervals are open then makes them disjoint. 
\end{proof}

The set $\mathcal{R}$ induces an interval order $P_\mathcal{R}$ over $X$, which is the partial order $\prec_{\mathcal{R}}$ where $x_i$ and $x_j$ are incomparable whenever $R_i$ and $R_j$ intersect, and where otherwise $x_i \prec_{\mathcal{R}} x_j$ whenever $b_i \leq a_j$. 
By Lemma~\ref{lem:disjointness}, any linear order $L$ that extends $P_\mathcal{R}$ must also extend $P_G$ and so $e(P_\mathcal{R}) \leq e(P_G)$. The number of linear extensions of an interval order is much easier to count. 
In fact, Cardinal, Fiorini, Joret, Jungers and Munro~\cite{cardinal_sorting_2013} and Van der Hoog, Kostityna, L\"{o}ffler and Speckmann~\cite{van_der_hoog_preprocessing_2019} already upper bound the number of linear extensions of an interval order.
We paraphrase their upper bound, to give a weaker version applicable to this paper. For a proof from first principles, see Section~\ref{sec:firstprinciples}.

\begin{restatable}[Lemma~1 in~\cite{van_der_hoog_preprocessing_2019} and Lemma 3.2 in~\cite{cardinal_sorting_2013}]{lemma}{interval}
    \label{lem:intervalorder}
    Let $\mathcal{R} = (R_1, \ldots R_n)$ be a set of $n$ open intervals in $[0, n]$ and let each interval have at least unit size.  Let $P_\mathcal{R}$ be its induced partial order.
    Then:
    \[
    \sum\limits_{i=1}^n \log( |R_i| )  \in O(\log e(P_\mathcal{R})). 
    \]
\end{restatable}

We are now ready to prove our main theorem:

\begin{theorem}
    Given a directed acyclic graph $G$ over $X$, inducing a partial order $P_G$, and an oracle $O_L$ whose queries specify a linear order $L$ that extends $P_G$, there exists an algorithm that uses linear space, $O(n + m + \log e(P_G))$ time and $O(\log e(P_G))$ oracle queries to output the sorted order of $X$. 
\end{theorem}

\begin{proof}
    Our algorithm runs in $O(n + m + k + \sum\limits_{i=1}^k \log d_i)$ time and uses $O(k + \sum\limits_{i=1}^k \log d_i)$ queries. 
    By Lemma~\ref{lem:kbound}, $k \in O(\log e(P_G))$. 
    The set $\mathcal{R}$ of Definition~\ref{def:intervals} is a set where each interval $R_i$ has at least unit size.
    By Lemma~\ref{lem:disjointness}, $e(P_\mathcal{R}) \leq e(P_G)$. 
    For $i \leq k$, the size $|R_i|$ must be at least $d_i$ since there are per construction at least $d_i$ vertices on the subpath from $\pi^*(p_i)$ to $\pi^*(q_j)$ in the embedding $E$. 
    It follows by Lemma~\ref{lem:intervalorder} that $\sum\limits_{i=1}^k \log d_i \in O(\log e(P_\mathcal{R}) ) \subseteq O(\log e(P_G))$. 
\end{proof}

\section{Deriving Lemma~\ref{lem:intervalorder} from first principles}
\label{sec:firstprinciples}

We consider the following statement:

\interval*

\noindent

While this lemma follows verbatim from Lemma 1 in~\cite{van_der_hoog_preprocessing_2019}    (which itself is weaker version of Lemma 3.2 in in~\cite{cardinal_sorting_2013}), we also give a proof from first principles. 
\newcommand{\Vol}{\operatorname{Vol}}
We reindex $\mathcal{R}$ so that $(R_1, R_2, \dots, R_m)$ is a maximum cardinality set of pairwise-disjoint intervals in $R$, sorted from left to right.  
For any $R_i \in \mathcal{R}$ we denote by $\operatorname{mid}(R_i)$ its centre. 
We define an $(n-m)$-dimensional polytope associated with $\mathcal{R}$.
$$
A = \Big\{x \in \mathbb{R}^n \Big| x_i = \operatorname{mid}(R_i) \text{ for } i \le m \text{ and } x_i \in R_i \text{ for } i > m \Big\}
$$
The volume of this polytope is $\Vol(A) =  \prod\limits_{i=m+1}^{n} |R_i|$.
Any $x = (x_1, \ldots, x_n) \in A$ with distinct coordinates defines a linear extension $L$ of $P_\mathcal{R}$ via the rule ``$i \prec j$ if and only if $x_i < x_j$''. We call the point $x$ a \emph{realisation of $L$}.
Observe that not all linear extensions of $P_\mathcal{R}$ have a realisation $x \in A$.

\begin{claim}
\label{claim:one}
Let $L$ be a linear extension of $P_R$. Let $A_L = \left\{x \in A \big| x \textnormal{ realises }L\right\}$, then
$$
\Vol(A_L) \le (2e)^{n-m}.
$$
\end{claim}
\begin{proof}
Consider the $m+1$ open intervals:

\[
I := (-0.5, \operatorname{mid}(R_1)), \ldots, (\operatorname{mid}(R_i), \operatorname{mid}(R_{i+1})), \dots,  (\operatorname{mid}(R_{m-1}), n + 0.5).
\]

For all $(i, j) \in [m] \times [m]$, the intervals $(R_i, R_j)$ are disjoint and have length at least $1$. It follows that each interval in $I$ 
has length at least $1$.
Suppose that $L$ cannot be realised, then $\Vol(A_L) = 0$. 
Otherwise, we observe that for any $x, x'$ realising $L$ and any $i \in \{m+1, ..., n\}$, an interval in $I$ contains $x_i$ if and only if it contains $x_i'$.
This allows us to say that an interval in $I$ is \emph{occupied under $L$} if it contains $x_i$ for any realisation $x$ of $L$ and any $i > m$.
There are at most $n-m$ occupied intervals, and therefore at least $m+1 - (n-m) = (2 m + 1 - n)$ non-occupied intervals in $I$.
Let $G_L$ be the union of the occupied intervals. Then $|G_L| \le |(-0.5, n+0.5)| - (2m+1 - n) = 2(n - m)$, as each non-occupied interval has length at least $1$. Let:
$$
B_L = \left\{y \in \mathbb{R}^{m} \times G_L^{n-m} \Big|  y_i = \operatorname{mid}(R_i) \text{ for } i \le m \text{ and } (x_j < x_k \Leftrightarrow y_j < y_k) \text{ for } j, k > m\right\}.
$$
Then $A_L \subseteq B_L$, hence
$$
\Vol(A_L) \le \Vol(B_L) = \frac{|G_L|^{n-m}}{(n-m)!} = \frac{(2(n-m))^{n-m}}{(n-m)!} \le (2e)^{n-m}. \qedhere
$$

\end{proof}
\begin{claim}
\label{claim:two}
$\sum\limits_{i=m+1}^{n} \ln |R_i| \le \ln e(P_R) + (1+\ln 2) (n-m)$
\end{claim}
\begin{proof}
We note that:
$$
\sum_{i=m+1}^{n} \ln |R_i| = \ln \Vol(A) \le \ln\sum_{L \supseteq P_R} \Vol(A_L).$$
We now use Claim~\ref{claim:one} to upper bond $\Vol(A_L)$ for all $L$ and the claim follows. 
\end{proof}

\begin{claim}
\label{claim:three}
$$
\sum_{i=1}^{m} \ln |R_i| \le (n-m)
$$
\end{claim}
\begin{proof}
As $R_1, \dots, R_m$ are pairwise disjoint intervals in $[0, n]$, $\sum_{i=1}^{m} |R_i| \le n$. 

By analysis, $\ln x \le x-1$, hence
$$
\sum_{i=1}^{m} \ln |R_i| \le  \sum_{i=1}^{m} (|R_i| - 1) =  (n - m). \qedhere
$$
\end{proof}
Combining Claim~\ref{claim:two} and Claim \ref{claim:three} yields
$$
\sum_{i=1}^{n} \log |R_i| = \frac{1}{\ln(2)} \cdot \sum_{i=1}^{n} \ln |R_i| \le \log e(P_R) + \Big(1+\frac{2}{\ln(2)}\Big) \cdot (n-m).
$$
We now apply Lemma~\ref{lem:kbound} of our paper to note that $n-m \le \log e(P_R)$ and Lemma 3 follows.

\bibliographystyle{plainurl}
\bibliography{refs}
\end{document}